\documentclass[11pt, a4paper]{article}

\usepackage{graphicx}
\usepackage{amsmath, amssymb, amsthm, amsfonts}
\usepackage{float, color}
\usepackage{ascmac, fancybox}
\usepackage{bm}

\def\Z{{\mathbb Z}}
\def\R{{\mathbb R}}  
\def\C{{\mathbb C}}

\def\X{{\mathcal{X}}}
\def\Y{{\mathcal{Y}}}
\def\Z{{\mathcal{Z}}}
\def\H{{\mathcal{H}}}
\def\L{{\mathcal{L}}}
\def\S{{\mathcal{S}}}
\def\M{{\mathcal{M}}}
\def\E{{\mathcal{E}}}

\def\a{{\alpha}}
\def\b{{\beta}}
\def\ga{{\gamma}}
\def\d{{\delta}}
\def\l{{\lambda}}
\def\m{{\mu}}

\def\r{{\rho}}
\def\s{{\sigma}}
\def\t{{\tau}}

\def\Tr{{\rm Tr}\,}

\newcommand{\bracket}[2]{\left \langle #1 \left | #2 \right\rangle\right.}

\newcommand{\B}{\hfill$\Box$\bigskip}

\newtheorem{theorem}    {Theorem}
\newtheorem{lemma}      [theorem]{Lemma}
\newtheorem{corollary}  [theorem]{Corollary}
\newtheorem{proposition}        [theorem]{Proposition}

\title{Information geometry of sandwiched R\'enyi $\a$-divergence}
\author{Kaito Takahashi and Akio Fujiwara%
\thanks{fujiwara@math.sci.osaka-u.ac.jp}\\
{Department of Mathematics, Osaka University}\\ 
{Toyonaka, Osaka 560-0043, Japan}\\
}
\date{}

\begin{document}
\maketitle

\begin{abstract}
Information geometrical structure $(g^{(D_\a)}, \nabla^{(D_\a)},\nabla^{(D_\a)*})$ induced from the sandwiched R\'enyi $\a$-divergence  
$D_\alpha(\rho\|\sigma):=\frac{1}{\alpha (\alpha-1)}\log\,\Tr \left(\s^{\frac{1-\a}{2\a}}\r\,\s^{\frac{1-\a}{2\a}}\right)^{\a}$ 
on a finite quantum state space $\S$ is studied. 
It is shown that the Riemannian metric $g^{(D_\a)}$ is monotone if and only if $\a\in(-\infty, -1]\cup [\frac{1}{2},\infty)$, 
and that the quantum statistical manifold $(\S, g^{(D_\a)}, \nabla^{(D_\a)},\nabla^{(D_\a)*})$ 
is dually flat if and only if $\a=1$.
\end{abstract}

\section{Introduction}

In his seminal paper, R\'enyi \cite{Renyi:1961} introduced a new class of information divergence now usually referred to as the {\em R\'enyi relative entropy} of order $\a$, where $\a$ is a positive number.
Recently, Wilde {\em et al.} \cite{WildeWY} and M\"uller-Lennert {\em et al.} \cite{MullerLennertDSFT} independently proposed an extension of the R\'enyi relative entropy to the quantum domain.
Let $\L(\H)$ and $\L_{\rm sa}(\H)$ denote the sets of linear operators and selfadjoint operators on a finite dimensional complex Hilbert space $\H$, 
and let $\L_+(\H)$ and $\L_{++}(\H)$ denote the subsets of $\L_{\rm sa}(\H)$ comprising positive operators and strictly positive operators. 
Given $\r, \s\in\L_+(\H)$ with $\r\neq 0$, let 
\begin{equation}\label{eqn:psi}
 \psi(\a):=\log\,\Tr \left(\s^{\frac{1-\a}{2\a}}\r\,\s^{\frac{1-\a}{2\a}}\right)^{\a}
\end{equation}
for $\a\in (0,\infty)$, with the convention that $\psi(\a):=\infty$ if $\a>1$ and $\ker \s \not\subset \ker \r$.
The first divided difference of $\psi$ at $\a=1$, i.e., 
\[ \frac{\psi(\a)-\psi(1)}{\a -1},\qquad (\a \neq 1)  \]
is called the {\em sandwiched R\'enyi relative entropy} \cite{WildeWY} or the {\em quantum R\'enyi divergence} \cite{MullerLennertDSFT}, and is denoted by $\tilde D_\alpha(\rho \|\sigma)$ in the present paper.
It is explicitly written as
\begin{equation}\label{eqn:Dtilde}
 \tilde D_\a(\r \|\s):=
 \frac{1}{\a -1}\log\,\Tr \left(\s^{\frac{1-\a}{2\a}}\r\,\s^{\frac{1-\a}{2\a}}\right)^{\a}
 -\frac{1}{\a -1} \log\,\Tr  \r.
\end{equation}
The sandwiched R\'enyi relative entropy is extended to $\a=1$ by continuity, to obtain the {\em Umegaki relative entropy}: 
\[
\tilde D_1(\r||\s)=\lim_{\a \to  1}  \tilde D_\a (\r||\s)=\Tr \!\left\{\r (\log \r -\log \s)\right\}
\]
with the convention that $\tilde D_1(\r||\s)=\infty$ if $\ker \s \not\subset \ker \r$.
The limiting cases $\a\downarrow 0$ and $\a\to \infty$ have also been studied in \cite{DattaL, AudenaertD} and \cite{MullerLennertDSFT}, respectively. 

The sandwiched R\'enyi relative entropy has several desirable properties: amongst others, it is monotone under completely positive trace preserving (CPTP) maps if $\a\ge \frac{1}{2}$ \cite{MullerLennertDSFT, WildeWY, Beigi,FrankL}. 
This property was successfully used in studying the strong converse properties of the channel capacity \cite{WildeWY, MosonyiO:2014} and quantum hypothesis testing problems \cite{MosonyiO:2013}. 

Now we confine our attention to the case when both $\r$ and $\s$ are faithful density operators that belong to the {\em quantum state space}:
\[ \S:=\S(\H):=\{\r\in \L_{++}(\H)\; | \; \Tr\r=1\}.  \]
In this case there is no difficulty in extending the quantities (\ref{eqn:psi}) and (\ref{eqn:Dtilde}) to the region $\a<0$. 
In order to motivate our study, let us assume for now that $\r$ and $\s$ commute.
Then the quantity (\ref{eqn:psi}) is reduced to $\psi(\a)=\log\Tr \r^\a\s^{1-\a}$: 
this is known as the potential function for the $\nabla^{(e)}$-geodesic connecting $\r$ and $\s$ in classical information geometry \cite{AmariNagaoka}, 
and is meaningful for all $\a\in\R$. 
On the other hand, 
\[ \tilde D_\a(\r \|\s)=\frac{1}{\a-1}\log\Tr \r^\a\s^{1-\a} \]
for $\a<0$ does not seem to be a reasonable measure of information \cite{Renyi:1961}, since it takes negative values. 
Meanwhile, there are also other types of divergence functions that have been found to be useful in classical information theory and statistics.
For example,  Csisz\'ar \cite{Csiszar} introduced a class of information divergence now usually referred to as the {\em Csisz\'ar $f$-divergence}, a version of which is written as
\[
 D^f(\r\|\s):=\Tr\{\s f(\r\s^{-1})\},
\]
where $f$ is a real-valued, strictly convex, smooth function on the set $\R_{++}$ of positive real numbers satisfying $f(1)=0$ and $f''(1)=1$. 
It is easily seen from Jensen's inequality that $D^f(\r\|\s)\ge 0$, and $D^f(\r\|\s)=0$ if and only if $\r=\s$. 
Now let us consider a family of functions
\[
 f^{[\a]}(t):=\frac{1}{\a(1-\a)} \left(1-t^{\a}\right)
\]
having a one-dimensional parameter $\a$ with $\a\notin\{0, 1\}$.
This family is known to play an important role in classical information geometry. 
For example, the corresponding $f^{[\a]}$-divergences are the well-known ``alpha-divergences'' 
\begin{equation}\label{eqn:Csiszar-f-divergence}
 D^{f^{[\a]}}(\r\|\s)=\frac{1}{\a(1-\a)}\left(1-\Tr \r^\a\s^{1-\a} \right),
\end{equation}
although the parametrization ``alpha'' is different from the standard one \cite{AmariNagaoka}. 
Now from the Taylor expansion of $\log(1+x)$, 
the divergence (\ref{eqn:Csiszar-f-divergence}) 
is related to the potential function $\psi(\a)=\log\Tr \r^\a\s^{1-\a}$ as 
\begin{equation}\label{eqn:divergenceExpansion}
 \psi(\a)=\a(\a-1) D^{f^{[\a]}}(\r\|\s)+O\left( D^{f^{[\a]}}(\r\|\s)^2 \right).
\end{equation}
In other words, the $f^{[\a]}$-divergence $D^{f^{[\a]}}(\r\|\s)$, provided it is small enough, is not very different from 
$\frac{1}{\a} \tilde D_\a(\r \|\s)$.
Note that the quantity $\frac{1}{\a} \tilde D_\a(\r \|\s)$ is nonnegative even for $\a<0$.

Motivated by the above consideration, we aim at investigating the ``rescaled'' sandwiched R\'enyi relative entropy:
\begin{equation}\label{eqn:Dalpha}
D_\a(\r\|\s):=
\frac{1}{\a (\a-1)}
\log\,\Tr \left(\s^{\frac{1-\a}{2\a}}\r\,\s^{\frac{1-\a}{2\a}}\right)^{\a}
\end{equation}
on the quantum state space $\S$ for $\a \in \R \backslash \{0,1\}$,
which shall be referred to as the {\em sandwiched R\'enyi $\a$-divergence} in the present paper.
It is continuously extended to $\a=1$ as 
\[ D_1(\r \|\s):=\lim_{\a \to 1}D_\a (\r \|\s)=\Tr \!\left\{\r (\log \r -\log \s)\right\}. \]
However, we note that, unless $\r$ and $\s$ commute, (\ref{eqn:Dalpha}) cannot be extended to $\a=0$ because $\lim_{\a \to 0}D_\a (\r \|\s)$ does not always exist 
(cf., Appendix \ref{app:Dalpha0}). 
To put it differently, the sandwiched R\'enyi $0$-divergence $D_0(\rho\|\sigma)$ is excluded on the quantum state space $\S$.
This fact makes a striking contrast to classical information geometry. 

A basic property of the sandwiched R\'enyi $\a$-divergence (\ref{eqn:Dalpha}) is the following: 
\begin{equation}\label{eqn:quasiContrast}
 D_\a(\r \| \s) \ge 0, \quad \mbox{and}\quad
 D_\a(\r\|\s)=0 \quad\mbox{if and only if}\quad \r=\s
\end{equation}
for $\r,\s\in\S$ and $\a\in\R\backslash\{0\}$ (cf., Appendix \ref{app:Dalpha}). 
This fact enables us to introduce an information geometric structure on the quantum state space $\S$ through Eguchi's method \cite{Eguchi}. 
Firstly, the Riemannian metric $g^{(D_\a)}$ is defined by
\begin{equation}\label{eqn:gD}
 g_\r^{(D_\a)}(X,Y):= \left. \phantom{\prod\!\!\!\!\!\!\!\!\!} D_\a((XY)_\r\|\s)\right|_{\s=\r}
 := \left. \phantom{\prod\!\!\!\!\!\!\!\!\!} XYD_\a(\r\|\s)\right|_{\s=\r}. 
\end{equation}
In the last side, the vector fields $X$ and $Y$ are regarded as acting only on $\r$. 
Secondly, a pair of affine connections $\nabla^{(D_\a)}$ and $\nabla^{(D_\a)*}$ are defined by
\begin{eqnarray}
 g_\r^{(D_\a)}(\nabla^{(D_\a)}_XY,Z) \!\! &:=&\!\! \left. \phantom{\prod\!\!\!\!\!\!\!\!\!} -D_\a((XY)_\r\|(Z)_\s)\right|_{\s=\r},
 \label{eqn:nablaD}\\
 g_\r^{(D_\a)}(\nabla^{(D_\a)*}_XY,Z) \!\! &:=&\!\! \left. \phantom{\prod\!\!\!\!\!\!\!\!\!} -D_\a((Z)_\r\|(XY)_\s)\right|_{\s=\r}.
 \label{eqn:nablaD*}
\end{eqnarray}
The right-hand sides of (\ref{eqn:nablaD}) and (\ref{eqn:nablaD*}) are understood in an analogous way to (\ref{eqn:gD}). 
Since 
\[
 \left. \phantom{\prod\!\!\!\!\!\!\!\!\!} D_\a((X)_\r\|\s)\right|_{\s=\r}=0
\]
for any vector field $X$, which is a consequence of (\ref{eqn:quasiContrast}), 
the metric $g_\r^{(D_\a)}$ is also written as
\[ g_\r^{(D_\a)}(X,Y)=\!\! \left. \phantom{\prod\!\!\!\!\!\!\!\!\!} -D_\a((X)_\r\|(Y)_\s)\right|_{\s=\r}.  \]
It is now straightforward to verify the duality: 
\begin{equation}\label{eqn:duality}
 X g^{(D_\a)}(Y,Z)=g^{(D_\a)}(\nabla^{(D_\a)}_XY,Z)+g^{(D_\a)}(Y,\nabla^{(D_\a)*}_XZ).
\end{equation}
This property plays an essential role in information geometry \cite{AmariNagaoka}.

A Riemannian metric $g$ on a quantum state space is called {\em monotone} \cite{Petz} if it satisfies
\begin{equation}\label{eqn:monotoneMetric}
 g_\r(X,X)\ge g_{\ga(\r)}(\ga_*X, \ga_*X)
\end{equation}
for all states $\r\in\S$, tangent vectors $X\in T_\r\S$, and CPTP maps $\ga:\L(\H)\to\L(\H')$, with $\ga_*$ denoting the differential of $\ga$. 
The monotonicity (\ref{eqn:monotoneMetric}) implies that the distinguishability of two nearby states, measured by the metric $g$, cannot be enhanced by any physical process $\ga$.
This is a fundamental requirement for information processing, and hence, characterizing monotone metrics is important in quantum information theory. 

The main result of the present paper is the following

\begin{theorem}\label{thm:main}
The induced Riemannian metric $g^{(D_\alpha)}$ is monotone under CPTP maps if and only if
$\a \in(-\infty, -1]\cup [\frac{1}{2}, \infty)$. 
\end{theorem}

As a by-product, we arrive at the following corollary, the latter part of which was first observed by numerical evaluation \cite{MullerLennertDSFT}.

\begin{corollary}
The sandwiched R\'enyi $\a$-divergence $D_\alpha  (\rho||\sigma)$ is not monotone under CPTP maps if $\a\in(-1,0)\cup (0,\frac{1}{2})$. 
Consequently, the original sandwiched R\'enyi relative entropy $\tilde D_\alpha  (\rho||\sigma)$
is not monotone if $\a\in(0,\frac{1}{2})$.
\end{corollary}

We also study the dualistic structure $(g^{(D_\a)},\nabla^{(D_\a)},\nabla^{(D_\a)*})$ on the quantum state space $\S$, and obtain the following 

\begin{theorem}\label{thm:flat}
The quantum statistical manifold $(\S, g^{(D_\a)},\nabla^{(D_\a)},\nabla^{(D_\a)*})$ is dually flat if and only if $\a=1$. 
\end{theorem}

The paper is organized  as follows.
In Section 2, we compute the metric $g^{(D_\alpha)}$, and prove Theorem \ref{thm:main}. 
In Section 3, we investigate the dualistic structure $(g^{(D_\a)},\nabla^{(D_\a)}, \nabla^{(D_\a)*})$ on a quantum state space $\S$, and prove Theorem \ref{thm:flat}.
Section 4 is devoted to concluding remarks. 
Some additional topics are discussed in Appendices A-D. 

\section{Proof of Theorem \ref{thm:main}}

In quantum information geometry, it is customary to use a pair of operator representations of tangent vectors 
called the m-representation and the e-representation \cite{AmariNagaoka}. 
The {\em m-representation} $X^{(m)}$ of a tangent vector $X\in T_\r\S$ at $\r\in\S$ is simply defined by 
\[ X^{(m)}:=X\r. \] 
In order to introduce an e-representation, on the other hand, we need to specify a continuous monotone function $f:\R_{++}\to\R_{++}$ satisfying 
$f(1)=1$ and $f(t)=tf(\frac{1}{t})$. 
Once such a function $f$ is given, we define the corresponding {\em e-representation} $X_{f}^{(e)}$ of $X\in T_\r\S$ by 
\[ X_{f}^{(e)}:=f(\Delta_\r)^{-1}\left\{(X\r)\r^{-1}\right\}, \]
where 
$\Delta_\r$ is the {\em modular operator} associated with $\r\in\S$ defined by
\[ \Delta_\r:\L(\H)\longrightarrow \L(\H):  A\longmapsto \r A \r^{-1}.  \]
A Riemannian metric $g$ is then given by the pairing
\[  
g_\r(X,Y)
=\Tr\!\left\{ X_{f}^{(e)} Y^{(m)} \right\} 
\]
between e- and m-representations.
According to Petz's theorem \cite{Petz}, the metric $g$ represented in this form is monotone if and only if the function $f$ is operator monotone.
Thus, to prove Theorem \ref{thm:main}, we first derive the defining function $f=f^{(D_\a)}$ from each $D_\a$ (Lemma \ref{lem:fAlpha}), and then verify that the function $f^{(D_\a)}$ is operator monotone if and only if $\a \in(-\infty, -1]\cup [\frac{1}{2}, \infty)$ (Lemma \ref{lem:fBeta}). 

\subsection{Computation of metric}

We first note that, for a power function $f(x)=x^\l$ with $x>0$ and $\l\in\R$, the directional derivative
\[
 Df(A)[B]:=\lim_{t \to 0} \frac{f(A+tB)-f(A)}{t}, \qquad
 (A\in\L_{++}(\H),\;B\in\L_{\rm sa}(\H))
\]
is given by
\begin{equation}\label{eqn:GateauxFormula1}
 Df(A)[B]=\l \int_0^1 dt \int_0^\infty ds\,  A^{\l t}(sI+A)^{-1}B(sI+A)^{-1} A^{\l (1-t)}, 
\end{equation}
where $I$ is the identity (cf., (\ref{eqn:GateauxEx3}) in Appendix \ref{app:Gateaux}).   
The following formula
\begin{equation}\label{eqn:GateauxFormula2}
 D(\Tr A^\l)[B]=\Tr \!\left\{(\l A^{\l-1})B\right\}
\end{equation}
for differentiation under the trace operation 
also follows from (\ref{eqn:GateauxFormula1}). 

\begin{lemma}\label{lem:fAlpha}
For each $\a\in\R\backslash\{0,1\}$, the metric $g^{(D_\a)}$ is represented in the form
\begin{equation}\label{eqn:fAlpha}
 g_\r^{(D_\a)}(X,Y)=\Tr\!\left\{ X_{f^{(D_\a)}}^{(e)} Y^{(m)} \right\},
\end{equation}
where
\[ 
 f^{(D_\a)}(t):=(\a -1) \frac{t^{\frac{1}{\a}}-1}{\,1-t^{\frac{1-\a}{\a}}\,}
\] 
with the convention that $f^{(D_\a)}(1):=\lim_{t\to 1}f^{(D_\a)}(t)=1$. 
\end{lemma}

\begin{proof}
Recall that the metric $g_\r^{(D_\a)}$ was defined by 
\[ g_\r^{(D_\a)}(X,Y)= \left. \phantom{\prod\!\!\!\!\!\!\!\!} XYD_\a(\r\|\s)\right|_{\s=\r}, \] 
where $X$ and $Y$ act only on $\r$.
Since $D_\a(\r\|\s)$ is written as
\[ D_\a(\r\|\s)=\frac{1}{\a(\a-1)} \log\Tr A^\a \]
where
\[ A:=\s^{\frac{1-\a}{2\a}} \r\, \s^{\frac{1-\a}{2\a}}, \]
we have
\begin{eqnarray*}
YD_\a(\r\|\s)
&=&
\frac{1}{\a(\a-1)} \frac{Y (\Tr A^\a)}{\Tr A^\a} \\
&=&
\frac{1}{\a(\a-1)} \frac{D(\Tr A^\a)[YA]}{\Tr A^\a} \\
&=&
\frac{1}{\a(\a-1)} \frac{\Tr \!\left\{(\a A^{\a-1}) (YA)\right\} }{\Tr A^\a} \\
&=&
\frac{1}{\a-1} \frac{\Tr \!\left\{A^{\a-1} B_Y\right\} }{\Tr A^\a}.
\end{eqnarray*}
Here 
\[ B_Y:=YA=\s^{\frac{1-\a}{2\a}} (Y\r)\, \s^{\frac{1-\a}{2\a}}, \]
and the formula (\ref{eqn:GateauxFormula2}) was used in the third equality.  
Consequently, 
\begin{eqnarray}
XYD_{\a}(\rho\|\sigma)
&=&
 \frac{1}{\a -1}\left[
 \frac{X\Tr \!\left\{A^{\a-1} B_Y\right\}}{\Tr A^\a}
 -\frac{\left(\Tr \!\left\{A^{\a-1} B_Y\right\}\right) \left(\Tr \!\left\{(\a A^{\a-1}) B_X\right\}\right)}{(\Tr A^\a)^2}
 \right] \nonumber \\
&=&
 \frac{1}{\a -1}
 \frac{\Tr \!\left\{(XA^{\a-1}) B_Y\right\}+\Tr\!\left\{A^{\a-1} C_{XY}\right\}}{\Tr A^\a}  \nonumber \\
&&\qquad
 -\frac{\a}{\a-1}
 \frac{\left(\Tr \!\left\{A^{\a-1} B_Y\right\}\right) \left(\Tr \!\left\{A^{\a-1} B_X\right\}\right)}{(\Tr A^\a)^2},
 \label{eqn:YXD}
\end{eqnarray}
where
\[ 
 B_X:=\s^{\frac{1-\a}{2\a}} (X\r)\, \s^{\frac{1-\a}{2\a}}, \qquad
 C_{XY}:=XB_Y=\s^{\frac{1-\a}{2\a}} (XY\r)\, \s^{\frac{1-\a}{2\a}}.
\]
Since
\begin{eqnarray*}
 &&\left. \phantom{\prod\!\!\!\!\!\!\!\!\!} \Tr A^\a\right|_{\s=\r}=\Tr\left(\r^{\frac{1}{\a}}\right)^\a
 	=\Tr\r=1, \\
 && \left. \phantom{\prod\!\!\!\!\!\!\!\!\!} \Tr \left\{A^{\a-1}B_X\right\}\right|_{\s=\r}
 	=\Tr\left\{\r^{\frac{\a-1}{\a}}\r^{\frac{1-\a}{2\a}} (X\r)\, \r^{\frac{1-\a}{2\a}}\right\}
	=\Tr (X\r) =X\Tr\r=0, \\
 &&\left. \phantom{\prod\!\!\!\!\!\!\!\!\!}\Tr \left\{A^{\a-1}C_{XY}\right\}\right|_{\s=\r}
 	=\Tr\left\{\r^{\frac{\a-1}{\a}}\r^{\frac{1-\a}{2\a}} (XY\r)\, \r^{\frac{1-\a}{2\a}}\right\}
	=\Tr (XY\r) =XY\Tr\r=0,
\end{eqnarray*}
we have
\begin{equation}\label{eqn:fAlpha1}
g_\r^{(D_\a)}(X,Y)
=\left. \phantom{\prod\!\!\!\!\!\!\!\!\!} XYD_\a(\r\|\s)\right|_{\s=\r}
=\left. \frac{1}{\a -1} \Tr \!\left\{(XA^{\a-1}) B_Y\right\} \right|_{\s=\r}. 
\end{equation}
Now we invoke the formula (\ref{eqn:GateauxFormula1}), with $\l=\a-1$, to obtain
\begin{eqnarray}
 XA^{\a-1}
 &=&Df(A)[XA] \nonumber\\
 &=&(\a-1) \int_0^1 dt \int_0^\infty ds \, A^{(\a-1)t}(sI+A)^{-1}B_X(sI+A)^{-1} A^{(\a-1)(1-t)}.
 \label{eqn:fAlpha2}
\end{eqnarray}
Combining (\ref{eqn:fAlpha1}) with (\ref{eqn:fAlpha2}), we have
\begin{eqnarray}
&& g_\r^{(D_\a)}(X,Y) \nonumber\\
&&\quad 
  =\Tr \left\{B_Y\left.\int_0^1 dt \int_0^\infty ds \, A^{(\a-1)t}(sI+A)^{-1}B_X(sI+A)^{-1} A^{(\a-1)(1-t)}\right\} \right|_{\s=\r} \nonumber \\
&&\quad
  =\Tr \left\{ \r^{\frac{1-\a}{2\a}}(Y\r) \r^{\frac{1-\a}{2\a}}
	\int_0^1 dt \right. \nonumber\\
&&\qquad\qquad
	\left.\times\int_0^\infty ds \left(\r^{\frac{\a-1}{\a}}\right)^t (sI+\r^{\frac{1}{\a}})^{-1}\r^{\frac{1-\a}{2\a}}(X\r)\r^{\frac{1-\a}{2\a}}(sI+\r^{\frac{1}{\a}})^{-1} \left(\r^{\frac{\a-1}{\a}}\right)^{(1-t)}
	\right\} \nonumber\\
&&\quad
  =\Tr \left\{\r^{\frac{1-\a}{\a}} (Y\r) \r^{\frac{1-\a}{\a}}
	\int_0^1 dt
	\int_0^\infty ds \left(\r^{\frac{\a-1}{\a}}\right)^t (sI+\r^{\frac{1}{\a}})^{-1}(X\r)(sI+\r^{\frac{1}{\a}})^{-1} \left(\r^{\frac{\a-1}{\a}}\right)^{(1-t)}
	\right\}. \nonumber\\
\label{eqn:fAlpha3}
\end{eqnarray}
Comparing (\ref{eqn:fAlpha3}) with (\ref{eqn:fAlpha}), we see that the e-representation $X^{(e)}_{f^{(D_\a)}}$ of $X$ is given by
\[
X^{(e)}_{f^{(D_\a)}}=
\r^{\frac{1-\a}{\a}}
	\int_0^1 dt
	\int_0^\infty ds \left(\r^{\frac{\a-1}{\a}}\right)^t (sI+\r^{\frac{1}{\a}})^{-1}(X\r)(sI+\r^{\frac{1}{\a}})^{-1} \left(\r^{\frac{\a-1}{\a}}\right)^{(1-t)}\r^{\frac{1-\a}{\a}}.
\]
In order to determine the function $f^{(D_\a)}$, we introduce an orthonormal basis $\{e_i\}_{1\le i\le n}$ of $\H$ comprising eigenvectors of $\r$ each corresponding to the eigenvalue $p_i$. 
Then
\begin{eqnarray*}
&&\bracket{e_i}{X^{(e)}_{f^{(D_\a)}}e_j} \\
&&\qquad=
p_i^{\frac{1-\a}{\a}}
	\int_0^1 dt
	\int_0^\infty ds \left(p_i^{\frac{\a-1}{\a}}\right)^t \left(s+p_i^{\frac{1}{\a}}\right)^{-1}
	\bracket{e_i}{(X\r)e_j}
	\left(s+p_j^{\frac{1}{\a}}\right)^{-1} \left(p_j^{\frac{\a-1}{\a}}\right)^{(1-t)}p_j^{\frac{1-\a}{\a}} \\
&&\qquad=
\bracket{e_i}{(X\r)e_j} (p_i p_j)^{\frac{1-\a}{\a}}
	\int_0^1 dt \left(p_i^{\frac{\a-1}{\a}}\right)^t \left(p_j^{\frac{\a-1}{\a}}\right)^{(1-t)} 
	\int_0^\infty ds\left(s+p_i^{\frac{1}{\a}}\right)^{-1} \left(s+p_j^{\frac{1}{\a}}\right)^{-1}.	
\end{eqnarray*}
Since $X^{(e)}_{f^{(D_\a)}}$ is continuous in $\r$,  we can assume without loss of generality that eigenvalues $p_i$ are all different. 
Then by using the formulae
\[
 \int_0^1 x^t y^{1-t}dt=\frac{x-y}{\log x-\log y}
\]
and
\[
\int_0^\infty\frac{ds}{(s+x)(s+y)}=\frac{\log x-\log y}{x-y}
\]
for $x\neq y$, we get
\begin{eqnarray*}
\bracket{e_i}{X^{(e)}_{f^{(D_\a)}}e_j}
&\!\!=\!\!&
\bracket{e_i}{(X\r)e_j} (p_i p_j)^{\frac{1-\a}{\a}}
	\times \frac{p_i^{\frac{\a-1}{\a}}-p_j^{\frac{\a-1}{\a}}}{\log p_i^{\frac{\a-1}{\a}}-\log p_j^{\frac{\a-1}{\a}}}
	\times \frac{\log p_i^{\frac{1}{\a}}-\log p_j^{\frac{1}{\a}}}{p_i^{\frac{1}{\a}}-p_j^{\frac{1}{\a}}} \\
&\!\!=\!\!&
\bracket{e_i}{(X\r)e_j} 
\frac{\displaystyle 1-\left(\frac{p_i}{p_j}\right)^{\frac{1-\a}{\a}}}
{\displaystyle (\a-1)\, p_j\left\{\left(\frac{p_i}{p_j}\right)^{\frac{1}{\a}}-1\right\}}\\
&\!\!=\!\!&
\frac{\bracket{e_i}{(X\r)e_j}}{\displaystyle p_j\, f^{(D_\a)}\left(\frac{p_i}{p_j}\right)},
\end{eqnarray*}
for all $i\neq j$, where
\[
 f^{(D_\a)}(t)=(\a -1) \frac{t^{\frac{1}{\a}}-1}{\,1-t^{\frac{1-\a}{\a}}\,}.
\]
This completes the proof. 
\end{proof}

Let us examine some special cases. 
When $\a=\frac{1}{2}$, the function $f^{(D_{{1}/{2}})}(t)=\frac{1+t}{2}$ corresponds to the SLD metric, 
and when $\a=-1$, the function $f^{(D_{-1})}(t)=\frac{2t}{1+t}$ corresponds to the real RLD metric. 
Furthermore, the limiting function $f^{(D_{1})}(t):=\lim_{\a\to 1}f^{(D_{\a})}(t)=\frac{t-1}{\log t}$ corresponds to the Bogoliubov metric: this is consistent with the fact that $D_{1}(\r\| \s):=\lim_{\a\to 1}D_{\a}(\r\| \s)$ is the Umegaki relative entropy. 
It is well known that these three functions are operator monotone. 
Incidentally, another limiting function 
$f^{(D_{\pm\infty})}(t):=\lim_{\a\to\pm\infty}f^{(D_{\a})}(t)=\frac{t\,\mathrm{log}t}{t-1}=t/f^{(D_{1})}(t)$ is also operator monotone. 

\subsection{Operator monotonicity}

In what follows, we change the parameter $\a$ into $\b:=\frac{1}{\a}$, and denote the corresponding function $f^{(D_\a)}(t)$ by $f_\b(t)$, i.e., 
\[
 f_\b (t):=\frac{\b-1}{\b} \frac{t^\b - 1}{t^{\b-1}-1}
\]
where $\b\notin\{0,1\}$. 
We extend this function to $\b=0$ and $1$ by continuity, to obtain
\begin{eqnarray*}
 f_0(t)&\!\!:=\!\!&\lim_{\b\to 0}  f_\b (t)=\frac{t \log t}{t-1}, \\ 
 f_1(t)&\!\!:=\!\!&\lim_{\b\to 1}  f_\b (t)=\frac{t-1}{\log t}.  
\end{eqnarray*}

\begin{lemma}\label{lem:fBeta}
The function $f_\b(t)$ is operator monotone if and only if $\b\in[-1, 2]$.
\end{lemma}

\begin{proof}
We first prove the `if' part%
\footnote{
After almost completing the paper, the authors became aware that the `if' part had been proved in \cite{Furuta:2008}. 
Our proof is slightly simpler.}. 
A key observation is the identity
\[
f_{\frac{1}{2}-\delta}(t)=\frac{t}{f_{\frac{1}{2}+\delta}(t)}
\]
for all $\d\in\R$, which is easily verified by direct computation. 
It follows that if $f_{\frac{1}{2}+\delta}(t)$ is operator monotone, so is $f_{\frac{1}{2}-\delta}(t)$.
It then suffices to prove that $f_\b(t)$ is operator monotone if $\b\in[\frac{1}{2}, 2]$. 
Firstly, operator monotonicity of $f_1(t)$ is well known. 
Secondly, for $\b\in[\frac{1}{2},1)$, let us set $\ga:=1-\frac{1}{\b}$, which satisfies $-1\le\ga<0$. 
Since $t\mapsto t^\ga$ is operator convex, its first divided difference 
\[
\frac{t^\ga-1}{t-1}
\]
at $t=1$ is operator monotone \cite[Theorem V.3.10]{Bhatia}.
Moreover, since $x \mapsto -\frac{1}{x}$ is operator monotone, so is the function
\[
h_\ga(t):=\ga\,\frac{t-1}{t^\ga-1}.
\]
Furthermore, since the function $t \mapsto t^\b$ is operator monotone, so is 
\[
 f_\b(t)=h_\ga(t^\b).
\]
Finally, for $\b\in (1,2]$, rewrite $f_\b(t)$ into
\[
f_\b(t)
=
\frac{\b-1}{\b} 
\left(t + \frac{t-1}{t^{\beta -1}-1}\right).
\] 
Since $\frac{\beta-1}{\beta} >0$ and the map $x\mapsto\frac{1}{x}$ is order reversing, 
it suffices to show that the function 
\[
 t\longmapsto \frac{t^{\b-1}-1}{t-1}
\]
is operator monotone decreasing, and this is true because the above function is the first divided difference of an operator concave function $t\mapsto t^{\b-1}$. 

We next prove the `only if' part. 
Suppose $f_\b(t)$ is operator monotone. Then it must satisfy the inequalities
\[
\frac{2t}{1+t} \leq f_\beta(t) \leq \frac{1+t}{2}
\]
for all $t>0$ \cite{Petz}. 
In particular, by letting $t=e$, we have
\[ 
\frac{2e}{1+e} \le
\frac{\b-1} {\b} \frac{e^\b -1}{e^{\b-1}-1}
\leq \frac{1+e}{2}.
\] 
We shall prove that these inequalities hold only when $\b\in[-1, 2]$. 
Since
\[
 \left.\frac{\b-1} {\b} \frac{e^\b -1}{e^{\b-1}-1} \right|_{\b=-1}
 =\frac{2e}{1+e},
\]
and
\[
 \left.\frac{\b-1} {\b} \frac{e^\b -1}{e^{\b-1}-1} \right|_{\b=2}
 =\frac{1+e}{2},
\]
it suffices to prove that the function
\[
\b \longmapsto \frac{\b-1}{\b}\frac{e^\b -1}{e^{\b-1} -1}
\]
is strictly increasing for $\b\le -1$ or $\b\ge 2$. 
Taking the logarithm, this is rephrased that the function
\[
\b \longmapsto h(\b) - h(\b -1),
\]
where
\[
h(\b):=\log \frac{e^\b -1 }{\b}, 
\]
is strictly increasing, or equivalently, 
\[ h'(\b) - h'(\b-1) \]
is positive for $\b\le -1$ or $\b\ge 2$.
In view of the mean value theorem, we prove this by showing that $h''(\b)>0$ for $\b\le -1$ or $\b\ge 1$.
Since
\[
h''(\b)=\frac{1}{\b^2}-\frac{1}{\displaystyle \left(e^{\frac{\b}{2}}-e^{-\frac{\b}{2}} \right)^2}, 
\]
it suffices to prove that
\begin{equation}\label{eqn:monotoneBound2}
 \tilde h(\b):=\frac{e^{\frac{\b}{2}}-e^{-\frac{\b}{2}}}{\b}>1
\end{equation}
for $\b\neq 0$. 
Since $\tilde h(-\b)=\tilde h(\b)$, $\lim_{\b \to 0} \tilde h(\b)= 1$, and the derivative
\[
 \tilde h'(\b)=\frac{e^{\frac{\b}{2}}(\b-2)+e^{-\frac{\b}{2}}(\b+2)}{2\b^2}
 =\frac{e^{-\frac{\b}{2}}}{2\b^2}\int_0^\b x\, (\b-x)\, e^x\, dx
\]
is positive for $\b>0$, 
the inequality (\ref{eqn:monotoneBound2}) is verified.
\end{proof}

\section{Structure of quantum statistical manifold $\S$} \label{Structure}

In this section, we study the dualistic structure $(g^{(D_\a)},\nabla^{(D_\a)},\nabla^{(D_\a)*})$ on the quantum statistical manifold $\S$. 
In a quite similar way to the derivation of (\ref{eqn:fAlpha1}), 
it is proved (cf., Appendix \ref{app:connection}) that 
the affine connections $\nabla^{(D_\a)}$ and $\nabla^{(D_\a)*}$ 
defined by (\ref{eqn:nablaD}) and (\ref{eqn:nablaD*})  are explicitly given by
\begin{eqnarray}
g_\r^{(D_\a)}(\nabla^{(D_\a)}_XY,Z) 
&\!\!=\!\!&
 \frac{1}{\a -1}\left( \left.\phantom{\prod}\!\!\!\!\!\!\!\! \Tr
  \left\{ (ZXA^{\a-1}) B_Y+(ZA^{\a-1}) C_{XY} \right\} \right|_{\s=\r} \right. \nonumber\\
&&\qquad\qquad
-\left. 
\Tr\left\{ (Y\r)\, Z \left[ \r^{\frac{1-\a}{2\a}} \left.\phantom{\prod}\!\!\!\!\!\!\!\!(XA^{\a-1})\right|_{\s=\r} \r ^ {\frac{1-\a}{2\a}}\right]\right\} \right),
\label{eqn:nabla}
\end{eqnarray}
and
\begin{eqnarray}
g_\r^{(D_\a)}(\nabla^{(D_\a)*}_XY,Z) 
&\!\!=\!\!&
 \frac{1}{\a -1}\left( \left.\phantom{\prod}\!\!\!\!\!\!\!\! \Tr
  \left\{ (XB_Z)(YA^{\a-1})-(YXA^{\a-1}) B_Z-(YA^{\a-1}) C_{XZ} \right\} \right|_{\s=\r} \right. \nonumber \\
&&\qquad\qquad\qquad
+\Tr\left\{ (Z\r)\, X \left[ \r^{\frac{1-\a}{2\a}} \left.\phantom{\prod}\!\!\!\!\!\!\!\!(YA^{\a-1})\right|_{\s=\r} \r ^ {\frac{1-\a}{2\a}}\right]\right\} \nonumber \\
&&\qquad\qquad\qquad
+\left. 
\Tr\left\{ (Z\r)\, Y \left[ \r^{\frac{1-\a}{2\a}} \left.\phantom{\prod}\!\!\!\!\!\!\!\!(XA^{\a-1})\right|_{\s=\r} \r ^ {\frac{1-\a}{2\a}}\right]\right\} \right).
\label{eqn:nablaStar}
\end{eqnarray} 
Now, if the quantum state space $\S$ is dually flat with respect to the dualistic structure $(g^{(D_\a)},\nabla^{(D_\a)},\nabla^{(D_\a)*})$, then there is a pair of affine coordinate systems that allows us a variety of information geometrical techniques on $\S$ \cite{AmariNagaoka}.
It is therefore interesting to ask which value of $\a$ makes $\S$ dually flat. 
The answer is given by Theorem \ref{thm:flat}:  the quantum statistical manifold $(\S, g^{(D_\a)},\nabla^{(D_\a)},\nabla^{(D_\a)*})$ is dually flat if and only if $\a=1$. 

\bigskip\noindent
{\em Proof of Theorem \ref{thm:flat}.}
When $\a=1$, the sandwiched R\'enyi $\a$-divergence is reduced to the Umegaki relative entropy 
$D_1(\r\|\s)=\Tr\{\r(\log\r-\log\s)\}$, 
and the dually flatness of $\S$ with respect to $(g^{(D_1)},\nabla^{(D_1)},\nabla^{(D_1)*})$ is well known \cite{AmariNagaoka}. 

To prove the necessity, let us take a submanifold $\M$ of $\S$ comprising commutative density operators, which can be regarded as the space of classical probability distributions. 
Since the sandwiched R\'enyi $\a$-divergence restricted to $\M$ is identical, up to the first order, to the classical ``alpha-divergence'' as  (\ref{eqn:divergenceExpansion}), 
the restricted metric $\left.g^{(D_\a)}\right|_\M$ is identical to the classical Fisher metric for all $\a$, and the restricted connections $\left.\nabla^{(D_\a)}\right|_\M$ and $\left.\nabla^{(D_\a)*}\right|_\M$ are the $(2\a-1)$- and the $(1-2\a)$-connections, respectively, in the standard terminology of classical information geometry. 
Consequently, $\M$ is dually flat if and only if $\a=0$ or $1$ \cite{AmariNagaoka}. 
Now suppose that $\S$ is dually flat with respect to the structure $(g^{(D_\a)},\nabla^{(D_\a)},\nabla^{(D_\a)*})$. 
Then the submanifold $\M$ is also dually flat with respect to the restricted structure $(\left.g^{(D_\a)}\right|_\M, \left.\nabla^{(D_\a)}\right|_\M, \left.\nabla^{(D_\a)*}\right|_\M)$. 
Since the sandwiched R\'enyi $0$-divergence is excluded on $\S$, as mentioned in Section 1,
the only remaining possibility is $\a=1$.
\B

A closely related question is this: Is there a triad $(\a,\b,\ga)$ of real numbers for which   
$(\S, g^{(D_\a)},\nabla^{(D_\b)},\nabla^{(D_\ga)})$ becomes dually flat? 
The answer is negative. 
In fact, 
the connections $\left.\nabla^{(D_\b)}\right|_\M$ and $\left.\nabla^{(D_\ga)}\right|_\M$ restricted to a commutative submanifold $\M$ are the $(2\b-1)$- and the $(2\ga-1)$-connections, respectively. 
If $(\S, g^{(D_\a)},\nabla^{(D_\b)},\nabla^{(D_\ga)})$ is dually flat, then the pair $(\b,\ga)$ must be either $(1,0)$ or $(0,1)$, as discussed above. 
Since the sandwiched R\'enyi $0$-divergence is excluded on $\S$, 
we conclude that there is no triad $(\a,\b,\ga)$ that makes   
$(\S, g^{(D_\a)},\nabla^{(D_\b)},\nabla^{(D_\ga)})$ dually flat. 

\section{Concluding remarks}

In the present paper, we studied information geometrical structure of the quantum state space $\S$ induced from the sandwiched R\'enyi $\a$-divergence $D_\a(\r\|\s)$, a variant of the quantum relative entropy recently proposed by Wilde {\em et al.} \cite{WildeWY} and M\"uller-Lennert {\em et al.} \cite{MullerLennertDSFT}.
We found that the induced Riemannian metric $g^{(D_\a)}$ is monotone if and only if $\a\in(-\infty, -1]\cup [\frac{1}{2}, \infty)$, and that the induced dualistic structure $(g^{(D_\a)}, \nabla^{(D_\a)},\nabla^{(D_\a)*})$ makes the quantum state space $\S$ dually flat if and only if $\a=1$. 

The result about the monotonicity of $g^{(D_\a)}$, which is consistent with the known monotonicity of $D_\a(\r\|\s)$ 
for $\a\in [\frac{1}{2}, \infty)$ \cite{FrankL}, strongly suggests that $D_\a(\r\|\s)$ might be monotone also for $\a\in(-\infty, -1]$.
This problem raises another interesting question about reconstructing $D_\a(\r\|\s)$ 
from a purely differential geometrical viewpoint. 
It is well known that the canonical divergence on a dually flat statistical manifold $(M, g, \nabla, \nabla^*)$ 
is reconstructed by integrating the metric $g$ along either $\nabla$ or $\nabla^*$-geodesic \cite{{AmariNagaoka},{AyAmari:2015}}. 
Unfortunately, this method is not applicable to our problem because the quantum statistical manifold $(S, g^{(D_\a)}, \nabla^{(D_\a)},\nabla^{(D_\a)*})$ is not dually flat unless $\a=1$.
If such a differential geometrical method of reconstructing a divergence function is successfully extended to non-flat statistical manifolds, then we may have a new, direct method of proving the monotonicity of a global quantity $D_\a(\r\|\s)$ from a local information $(g^{(D_\a)}, \nabla^{(D_\a)},\nabla^{(D_\a)*})$ on the quantum statistical manifold $\S$.

\section*{Acknowledgments}

The authors are grateful to Professor Tomohiro Ogawa for helpful discussions. 
The present study was supported by JSPS KAKENHI Grant Number JP22340019.

\appendix
\section*{Appendices}

\section{Non-extendibility of $D_\a$ to $\a=0$}\label{app:Dalpha0}

The following Proposition is a special case of \cite[Lemma 1]{AudenaertD}.

\begin{proposition}\label{prop:limitZero}
For all $\r, \s\in\S$,
\[
 \lim_{\a\to 0} \left(\s^{\frac{1-\a}{2\a}}\r\,\s^{\frac{1-\a}{2\a}}\right)^{\a}=\s.
\]
\end{proposition}

\begin{proof}
Since $\dim\H<\infty$ and $\r>0$, there are positive numbers $\l$ and $\m$ that satisfy $\l I \le \r \le \m I$.
This entails that 
\[
  \l \s^{\frac{1-\a}{\a}} \le \s^{\frac{1-\a}{2\a}}\r\,\s^{\frac{1-\a}{2\a}} \le \m \s^{\frac{1-\a}{\a}}.
\]
For $\a\in(0,1)$, the function $f(t)=t^\a$ with $t>0$ is operator monotone, and 
\begin{equation}\label{eqn:boundSandwich1}
  \l^\a \s^{1-\a} \le  \left(\s^{\frac{1-\a}{2\a}}\r\,\s^{\frac{1-\a}{2\a}}\right)^{\a} \le \m^\a \s^{1-\a}.
\end{equation}
For $\a\in(-1,0)$, on the other hand, the  function $f(t)=t^\a$ with $t>0$ is operator monotone decreasing, and 
\begin{equation}\label{eqn:boundSandwich2}
  \m^\a \s^{1-\a} \le  \left(\s^{\frac{1-\a}{2\a}}\r\,\s^{\frac{1-\a}{2\a}}\right)^{\a} \le \l^\a \s^{1-\a}.
\end{equation}
Taking the limit $\a\downarrow 0$ in (\ref{eqn:boundSandwich1}), and $\a\uparrow 0$ in  (\ref{eqn:boundSandwich2}), we have the assertion. 
\end{proof}

In view of Proposition \ref{prop:limitZero}, 
as well as the evaluation
\[ -\log\m-H(\s) \le \liminf_{\a\to 0} D_\a(\r\|\s) \le \limsup_{\a\to 0} D_\a(\r\|\s) \le -\log\l-H(\s), \]
where $H(\s)$ is the von Neumann entropy, which follows from (\ref{eqn:boundSandwich1}) and (\ref{eqn:boundSandwich2}),
it is natural to expect that $D_\a(\r\|\s)$ could be continuously extended to $\a=0$. 
In reality, it is in general untrue, as the following example shows.

Let $\H=\C^2$ and let
\[
 \r=\left[\array{cc} {1}/{2} & {1}/{4} \\ {1}/{4} & {1}/{2} \endarray\right],
 \qquad
 \s=\left[\array{cc} {3}/{4} & 0\\ 0 & {1}/{4} \endarray\right].
\]
The eigenvalues of $\s^{\frac{1-\a}{2\a}}\r\,\s^{\frac{1-\a}{2\a}}$ are
\[
 \frac{1}{3\cdot 4^{{1}/{\a}}} \left(3+3^{{1}/{\a}} \pm \sqrt{9-3\cdot 3^{{1}/{\a}}+9^{{1}/{\a}}} \right),
\]
and $\Tr \left(\s^{\frac{1-\a}{2\a}}\r\,\s^{\frac{1-\a}{2\a}}\right)^{\a}$ is given by
\[
\frac{1}{3^\a\cdot 4} 
 \left\{\left(3+3^{{1}/{\a}} + \sqrt{9-3\cdot 3^{{1}/{\a}}+9^{{1}/{\a}}} \right)^\a
 +\left(3+3^{{1}/{\a}} - \sqrt{9-3\cdot 3^{{1}/{\a}}+9^{{1}/{\a}}} \right)^\a \right\}.
\]
By direct computation, we obtain
\[ 
 \lim_{\a\downarrow 0}  D_\a(\r\|\s)=\frac{1}{2}\log \left(\frac{3}{2}\right)
\]
and
\[
 \lim_{\a\uparrow 0}  D_\a(\r\|\s)=\frac{1}{2}\log 2, 
\]
proving that 
\[ \lim_{\a\downarrow 0}  D_\a(\r\|\s) \neq \lim_{\a\uparrow 0}  D_\a(\r\|\s).\] 

\section{A basic property of $D_\a$}\label{app:Dalpha}

The following Proposition
is an extension of the results by Wilde {\em et al.} \cite[Corollaries 7 and 8]{WildeWY} for $1<\a\le 2$, M\"uller-Lennert {\em et al.} \cite[Theorem 3]{MullerLennertDSFT} for $\a\ge\frac{1}{2}$, and Beigi \cite[Theorem 5]{Beigi} for $\a>0$.

\begin{proposition}
$D_\a(\r\|\s)\ge 0$ for all $\a\in\R\backslash\{0\}$ and $\r, \s\in\S$. 
Moreover, the equality holds if and only if $\r=\s$.
\end{proposition}

\begin{proof}
Since the case $\a>0$ has been treated in \cite{Beigi}, we shall concentrate on the case $\a<0$; 
however, we note that our method presented here is also applicable to the case $\a>0$. 

Let $\s=\sum_i s_i E_i$ be the spectral decomposition, where $\{s_i\}_i$ are distinct eigenvalues of $\s$ and $\{E_i\}_i$ are the corresponding projection operators. 
The {\em pinching} operation $\E_\s:\L(\H)\to\L(\H)$ associated with the state $\s$ is defined by
\[
  \E_\s(A):=\sum_i E_i A E_i.
\]
The pinching $\E_\s$ sends a state $\r$ to a state $\E_\s(\r)$ that commutes with $\s$. 
To prove the first part of the claim, it suffices to show that $D_\a(\r\|\s)\ge D_\a(\E_\s(\r)\|\s)$ for all $\a\in\R\backslash\{0\}$, since the positivity $D_\a(\E_\s(\r)\|\s)\ge 0$ is well known in classical information theory. 
Due to \cite[Problem II.5.5]{Bhatia}, for any Hermitian matrix $A$, the vector $\l(\E_\s(A))$ comprising eigenvalues of $\E_\s(A)$ is majorised by the vector $\l(A)$ comprising eigenvalues of $A$; in symbol, $\l(\E_\s(A))\prec \l(A)$. 
It follows that $\Tr f(\E_\s(A))\le \Tr f(A)$ for any convex function $f$ (cf., (\ref{eqn:convex2}) below). 
Applying this result to $A=\s^{\frac{1-\a}{2\a}}\r\,\s^{\frac{1-\a}{2\a}}$ and $f(t)=t^\a$ with $t>0$, which is convex for $\a<0$, we have
\[
 \Tr \left(\E_\s(\s^{\frac{1-\a}{2\a}}\r\,\s^{\frac{1-\a}{2\a}})\right)^{\a}
 \le 
 \Tr \left(\s^{\frac{1-\a}{2\a}}\,\r\,\s^{\frac{1-\a}{2\a}}\right)^\a.
\]
Taking the logarithm of both sides, dividing them by $\a(\a-1)$, which is positive for $\a<0$, and noting that $\E_\s(\s^{\frac{1-\a}{2\a}}\r\,\s^{\frac{1-\a}{2\a}})=\s^{\frac{1-\a}{2\a}} \E_\s(\r)\,\s^{\frac{1-\a}{2\a}}$, we have $D_\a(\E_\s(\r)\|\s)\le D_\a(\r\|\s)$. 

Let us proceed to the second part. 
The `if' part is obvious. We show the `only if' part.
Since $D_\a(\r\|\s)\ge D_\a(\E_\s(\r)\|\s)\ge 0$,  the condition $D_\a(\r\|\s)=0$ leads to a series of equalities $D_\a(\r\|\s)=D_\a(\E_\s(\r)\|\s)$ and $D_\a(\E_\s(\r)\|\s)=0$. 
The latter equality implies that $\E_\s(\r)=\s$.
The former equality is equivalent to 
\[
 \Tr f(A)=\Tr f(\E_\s(A)),
\]
where $A=\s^{\frac{1-\a}{2\a}}\r\,\s^{\frac{1-\a}{2\a}}$ and $f(t)=t^\a$.
Since 
$f(t)$ is strictly convex in $t>0$ for $\a<0$, Lemma \ref{lem:convex} below shows that $\l^\downarrow(A)=\l^\downarrow(\E_\s(A))$. 
Here $\l^\downarrow(A)$ denotes the vector comprising eigenvalues of $A$ arranged in the decreasing order. 
It then follows from Lemma \ref{lem:pinching} below that $A=\E_\s(A)$, or equivalently, $\r=\E_\s(\r)$. 
Putting these results together, we have $\r=\E_\s(\r)=\s$. 
\end{proof}

\begin{lemma}\label{lem:convex}
Let $A$ and $B$ be strictly positive $n\times n$ Hermitian matrices, and let $f:\R_{++}\to\R$ be a strictly convex function. 
If $\l(A)\succ\l(B)$ and $\Tr f(A)=\Tr f(B)$, then $\l^\downarrow(A)=\l^\downarrow(B)$. 
\end{lemma}

\begin{proof}
Let us denote the eigenvalues of $A$ and $B$ explicitly as follows: 
\[ \l^\downarrow(A)=(\l_1,\l_2,\dots,\l_n), \qquad \l^\downarrow(B)=(\m_1,\m_2,\dots,\m_n), \]
where $\l_1\ge\l_2\ge\cdots\ge\l_n >0$ and $\m_1\ge\m_2\ge\cdots\ge\m_n >0$. 
Furthermore, let $(\hat\l_1,\hat\l_2,\dots, \hat\l_r)$ and $(\hat\m_1,\hat\m_2,\dots, \hat\m_s)$ be the lists of distinct eigenvalues of $A$ and $B$, respectively, labeled in the decreasing order, so that $\hat\l_1>\hat\l_2>\dots>\hat\l_r$ and 
$\hat\m_1>\hat\m_2>\dots>\hat\m_s$. 
In order to handle the multiplicity of eigenvalues, we introduce
the subsets 
\[
 I_\a:=\{i\,|\, \l_i=\hat\l_\a\},\qquad J_\b:=\{j\,|\, \m_j=\hat\m_\b\}
\]
of indices, each corresponding to distinct eigenvalue $\hat\l_\a$ or $\hat\m_\b$.
Then the set $\{1,2,\dots,n\}$ is decomposed into disjoint unions of $\{I_\a\}_{1\le\a\le r}$ and \{$J_\b\}_{1\le\b\le s}$ as follows:
\[
 \{1,2,\dots,n\}=\bigsqcup_{\a=1}^r I_\a=\bigsqcup_{\b=1}^s J_\b.
\]
We shall show that $r=s$, and that $\hat\l_\a=\hat\m_\a$ and $I_\a=J_\a$ for each $\a=1,\dots,r$. 

Since $\l(A)\succ\l(B)$, there is a doubly stochastic  $n\times n$ matrix $Q=[Q_{ji}]$ that satisfies
\begin{equation}\label{eqn:convex0}
 \m_j=\sum_{i=1}^n Q_{ji} \l_i,\qquad (j=1,\dots,n).
\end{equation}
Thus, for a convex function $f$, 
\begin{equation}\label{eqn:convex1}
 f(\m_j)=f\left(\sum_{i=1}^n Q_{ji} \l_i \right)\le \sum_{i=1}^n Q_{ji} f(\l_i)
\end{equation}
and 
\begin{equation}\label{eqn:convex2}
 \sum_{j=1}^n  f(\m_j) \le \sum_{i=1}^n  f(\l_i).
\end{equation}
If there is some $j$ for which the inequality (\ref{eqn:convex1}) becomes strict, then the inequality (\ref{eqn:convex2}) also becomes strict. 
Thus the condition $\Tr f(B)=\Tr f(A)$, which amounts to demanding equality in (\ref{eqn:convex2}), leads us to equality in  (\ref{eqn:convex1}) for all $j=1,\dots,n$.
Since $f$ is strictly convex, equality in (\ref{eqn:convex1}) holds if and only if there is an $\a\in\{1,\dots,r\}$ such that the support set $\{i\,|\, Q_{ji}>0\}$ is a subset of $I_\a$, i.e., 
\begin{equation}\label{eqn:convex3}
  \l_i=\hat\l_\a \quad\mbox{for all}\quad i \in \{k\,|\, Q_{jk}>0\}.
\end{equation}
Combining (\ref{eqn:convex3}) with (\ref{eqn:convex0}), we also have
\begin{equation}\label{eqn:convex4}
 \m_j=\hat\l_\a.
\end{equation}
Put differently, for each $j\in\{1,\dots,n\}$, there is a unique $\a$ that satisfies (\ref{eqn:convex3}) and (\ref{eqn:convex4}), and this correspondence defines a map $\Gamma: j\mapsto \a$. 

Given $\b\in\{1,\dots,s\}$, let us choose $j_1,j_2\in J_\b$ arbitrarily. Then we see from (\ref{eqn:convex4}) that
\[
 \hat\l_{\Gamma(j_1)}=\m_{j_1}=\hat\m_\b=\m_{j_2}=\hat\l_{\Gamma(j_2)}. 
\]
Consequently, $\Gamma(j_1)=\Gamma(j_2)$ for any $j_1,j_2\in J_\b$ and $\b\in\{1,\dots,s\}$.
This implies that $\Gamma$ naturally induces an injective 
map $\tilde\Gamma: \b\mapsto\a$ for which $\hat\m_\b=\hat\l_\a$. 
In particular, we must have $s\le r$.

Now, the above construction shows that
$Q_{ji}>0$ only if $(j,i)\in J_\b\times I_\a$ with $\a=\tilde\Gamma(\b)$. 
As a consequence, for any pair $(\a,\b)$ satisfying $\a=\tilde\Gamma(\b)$, 
\[
 |J_\b| \cdot \hat \m_\b
 =\sum_{j\in J_\b} \m_j 
 =\sum_{j\in J_\b} \left(\sum_{i\in I_\a} Q_{ji}\l_i\right)
 =\sum_{i\in I_\a}\l_i
 =|I_\a| \cdot \hat \l_\a.
\]
Since $\hat\m_\b=\hat\l_\a>0$, we have $|J_\b|=|I_\a|$.
This relation further concludes that $\tilde\Gamma$ is surjective; since otherwise $s<r$ from the injectivity of $\Gamma$, and
\[
 n=\sum_{\b=1}^s |J_\b|=\sum_{\b=1}^s |I_{\tilde\Gamma(\b)}|<\sum_{\a=1}^r |I_\a|=n,
\]
which is a contradiction. 

In summary, $\tilde\Gamma$ is bijective (in fact, the identity map), $s=r$, 
and  $\hat\m_\b=\hat\l_\b$ and $J_\b=I_\b$ for each $\b=1,\dots,s\, (=r)$. 
Consequently, we have $\l^\downarrow(A)=\l^\downarrow(B)$. 
\end{proof}

\begin{lemma}\label{lem:pinching}
Let $\E_\s$ be the pinching operation associated with a state $\s\in\S(\C^n)$.
For an $n\times n$ Hermitian matrix $A$, the following conditions are equivalent.
\begin{itemize}
\item[{\rm (i)}] $\l^\downarrow(A)=\l^\downarrow(\E_\s(A))$.
\item[{\rm (ii)}] $A=\E_\s(A)$.
\end{itemize}
\end{lemma}

\begin{proof}
(ii) $\Rightarrow$ (i) is trivial. We show (i) $\Rightarrow$ (ii). 
The characteristic polynomial $\varphi_A(x):=\det(x I-A)$ of $A=[a_{ij}]$ is expanded in $x$ as
\[
 \varphi_A(x)=x^n-x^{n-1}\left(\sum_i a_{ii}\right)+x^{n-2}\left(\sum_{i<j} (a_{ii}a_{jj}-|a_{ij}|^2)\right)+\cdots+(-1)^n\det A.
\]
Therefore
\begin{equation}\label{eqn:pinching}
 \varphi_A(x)-\varphi_{\E_\s(A)}(x)
 =x^{n-2}\left(-\sum_{(i,j)\in\Delta} |a_{ij}|^2\right)+\cdots,
\end{equation}
where
\[
 \Delta:=\{(i,j)\,|\,\mbox{$i<j$, and the element $a_{ij}$ is forced to be zero by the pinching $\E_\s$} \}.
\]
Since (i) implies $\varphi_A(x)=\varphi_{\E_\s(A)}(x)$, the coefficient of $x^{n-2}$ in (\ref{eqn:pinching}) must vanish. 
As a consequence, we have $a_{ij}=0$ for all $(i,j)\in\Delta$, proving that $A=\E_\s(A)$. 
\end{proof}

\section{Differential calculus}\label{app:Gateaux}

In calculating directional derivatives of functions on $\L(\H)$, knowledge about the G\^ateaux derivative is useful \cite{Bhatia}. 
Let $\X$ and $\Y$ be real Banach spaces and let $U$ be an open subset of $\X$. 
A continuous map $f: U \rightarrow \Y$ is said to be {\em G\^ateaux differentiable} at $x\in U$ if, for every $v\in \X$, the limit 
\[
 Df(x)[v]:=\lim_{t \to 0} \frac{f(x+tv)-f(u)}{t}
\]
exists in $\Y$. 
The quantity $Df(x)[v]$ is called the {\em G\^ateaux derivative} of $f$ at $x$ in the direction $v$. 
If $f$ is G\^ateaux differentiable at every point of $U$, we say that $f$ is G\^ateaux differentiable on $U$. 
The following basic properties of the G\^ateaux derivative are useful in applications.  

\begin{itemize}
\item[(I)] (Chain rule) 
If two maps $f : U \rightarrow \Y$ and $g:\Y \rightarrow \Z$ are G\^ateaux differentiable,
then their composition $g\circ f$ is also G\^ateaux differentiable, and
\[
D(g \circ f)(x)[v]=Dg(f(x))[Df(x)[v]] 
\]
holds for all $x\in U$ and $v \in \X$.

\item[(II)] (Product rule) 
Let $\t$ be a bilinear map from the product of two Banach spaces $\Y_1$ and $\Y_2$ into $\Z$.
If two maps $f: U \rightarrow \Y_1$ and $g: U \rightarrow \Y_2$ are G\^ateaux differentiable, then $\t(f,g)(x):= \t(f(x),g(x))$ is also G\^ateaux differentiable, and 
\[
D(\t(f,g))(x)[v]=\t(Df(x)[v],g(x))+\t(f(x),Dg(x)[v])
\]
holds for all $x\in U$ and $v \in \X$.
As a special case, let $\Y_1=\Y_2:=\L(\H)$ and let $\t$ be the usual product of two operators denoted by $\cdot$. 
Then we obtain
\[
D(f\cdot g)(x)[v]=Df(x)[v] \cdot g(x) + f(x)\cdot Dg(x)[v]
\]
for all $x\in U$ and $v \in \X$.
\end{itemize}

Now we derive some formulae for the G\^ateaux derivative. 
Firstly, let $f(A)=e^A$ with $A\in\L_{\rm sa}(\H)$. Then 
\begin{equation}\label{eqn:GateauxEx1}
Df(A)[B]=\int_0 ^1 e^{(1-t)A}Be^{tA}dt
\end{equation}
for all $A, B\in \L_{\rm sa}(\H)$. 
In fact, integrating the identity
\[
 \frac{d}{dt} \left\{e^{-tA}e^{t(A+B)} \right\}=e^{-tA}Be^{t(A+B)}
\]
and operating $e^{A}$ from the left, we have the {\em Dyson expansion}: 
\[
e^{A+B}-e^{A}=\int_0 ^1 e^{(1-t)A}Be^{t(A+B)}dt.
\]
Replacing $B$ in the above formula with $uB$, where $u\in\R$, we obtain
\begin{eqnarray*}
\lim_{u \to 0}\frac{e^{(A+uB)}-e^A}{u}
&=&\lim_{u \to 0}\int_0^1 e^{(1-t)A}Be^{t(A+uB)}dt \\
&=&\int_0 ^1 e^{(1-t)A}Be^{tA}dt.
\end{eqnarray*}

Secondly, let $f(A)=\log A$ with $A\in\L_{++}(\H)$. Then 
\begin{equation}\label{eqn:GateauxEx2}
Df(A)[B]=\int_0^\infty (s I +A)^{-1}B(s I+A)^{-1}ds
\end{equation}
for all $A\in\L_{++}(\H)$ and $B\in \L_{\rm sa}(\H)$. 
In fact, using the integral representation
\[
\log x=\int_0^{\infty}\left( \frac{1}{s+1}-\frac{1}{s+x} \right) ds, 
\]
we obtain
\begin{eqnarray*}
&&\lim_{u \to 0} 
\frac{\mathrm{log}(A+uB)-\mathrm{log}A}{u} \\
&&\qquad=
\lim_{u \to 0} \int_0^\infty 
\frac{(s+1)^{-1}I-(s I+A+uB)^{-1}-(s+1)^{-1}I+(s I+A)^{-1}}{u}ds \\
&&\qquad=
\lim_{u \to 0} \int_0^\infty 
\frac{(s I+A)^{-1}-(s I+A+uB)^{-1}}{u}ds \\
&&\qquad=
\int_0^\infty 
(s I+A)^{-1}B(s I+A)^{-1}ds.
\end{eqnarray*}
In the last equality, we used the resolvent identity: 
\[
(s I+Q)^{-1}-(s I+P)^{-1}=(s I+P)^{-1} (P-Q) (s I+Q)^{-1}.
\]

Finally, let $f(A)=A^\l$ with $\l\in\R$ and $A\in\L_{++}(\H)$. 
Then
\begin{equation}\label{eqn:GateauxEx3}
Df(A)[B]=\l \int_0^1 dt \int_0^\infty ds  \, A^{(1-t)\l}(sI+A)^{-1}B(sI+A)^{-1} A^{t\l}
\end{equation}
for all $A\in\L_{++}(\H)$ and $B\in \L_{\rm sa}(\H)$. 
In fact, since $f(A)=e^{\l\log A}=h(g(A))$, where $g(x)=\l\log x$ and $h(x)=e^x$, the chain rule yields
\begin{eqnarray*}
Df(A)[B]
&=&Dh(g(A))[Dg(A)[B]] \\
&=&Dh(g(A))\left[\l\int_0^\infty (s I +A)^{-1}B(s I+A)^{-1}ds \right] \\
&=&\int_0^1 e^{(1-t) (\l \log A)} \left[\l\int_0^\infty (s I +A)^{-1}B(s I+A)^{-1}ds \right] e^{t (\l \log A)}\,dt.
\end{eqnarray*}
In the second and the third equalities, we used (\ref{eqn:GateauxEx2}) and (\ref{eqn:GateauxEx1}), respectively.

\section{Computation of affine connections} \label{app:connection}

Let us derive the formulae (\ref{eqn:nabla}) and (\ref{eqn:nablaStar}) for the affine connections $\nabla^{(D_\alpha)}$ and $\nabla^{(D_\alpha)*}$. 
Since
\begin{eqnarray*}
 \left. \phantom{\prod\!\!\!\!\!\!\!\!\!} Xg^{(D_\a)}(Y,Z)\right|_{\r}
 &\!\!=\!\!&
 X_\r\left\{\left. \phantom{\prod\!\!\!\!\!\!\!\!\!} D_\a((YZ)_\r\|\s)\right|_{\s=\r}\right\} \\
 &\!\!=\!\!&
 \left. \phantom{\prod\!\!\!\!\!\!\!\!\!} D_\a((XYZ)_\r\|\s)\right|_{\s=\r}
 + \left. \phantom{\prod\!\!\!\!\!\!\!\!\!} D_\a((YZ)_\r\|(X)_\s)\right|_{\s=\r} \\
 &\!\!=\!\!&
  \left. \phantom{\prod\!\!\!\!\!\!\!\!\!} D_\a((XYZ)_\r\|\s)\right|_{\s=\r}
 - g_\r^{(D_\a)}(\nabla^{(D_\a)}_YZ,X),
\end{eqnarray*}
we have
\begin{equation}\label{eqn:nablaApp}
g_\r^{(D_\a)}(\nabla^{(D_\a)}_XY,Z)=
\left. \phantom{\prod\!\!\!\!\!\!\!\!\!} D_\a((ZXY)_\r\|\s)\right|_{\s=\r} 
- \left. \phantom{\prod\!\!\!\!\!\!\!\!\!} Z g^{(D_\alpha)}(X,Y)\right|_{\r}. 
\end{equation}
On the other hand, due to the duality (\ref{eqn:duality}), we have
\begin{eqnarray}
 g_\r^{(D_\a)}(\nabla^{(D_\a)*}_XY,Z)
 &\!\!=\!\!&
 \left. \phantom{\prod\!\!\!\!\!\!\!\!\!} Xg^{(D_\a)}(Y,Z)\right|_{\r}
 -g_\r^{(D_\a)}(\nabla^{(D_\a)}_XZ,Y)  \nonumber \\
 &\!\!=\!\!&
  \left. \phantom{\prod\!\!\!\!\!\!\!\!\!} Xg^{(D_\a)}(Y,Z)\right|_{\r}
 +\left. \phantom{\prod\!\!\!\!\!\!\!\!\!} Y g^{(D_\alpha)}(X,Z)\right|_{\r}
 -\left. \phantom{\prod\!\!\!\!\!\!\!\!\!} D_\a((YXZ)_\r\|\s)\right|_{\s=\r}. 
 \label{eqn:nablaStarApp}
\end{eqnarray}
Equations (\ref{eqn:nablaApp}) and (\ref{eqn:nablaStarApp}) imply that computing the affine connections $\nabla^{(D_\alpha)}$ and $\nabla^{(D_\alpha)*}$ is reduced to computing $D_\a((XYZ)_\r\|\s)|_{\s=\r}$ and $X g^{(D_\a)}(Y,Z)$.

We first compute $D_\a((XYZ)_\r\|\s)|_{\s=\r}$. 
We see from (\ref{eqn:YXD}) that
\begin{eqnarray*} 
D_\a((YZ)_\r\|\s)
&\!\!=\!\!&
 \frac{1}{\a -1}
 \frac{\Tr \!\left\{(YA^{\a-1}) B_Z\right\}+\Tr\!\left\{A^{\a-1} C_{YZ}\right\}}{\Tr A^\a} \\
 &&
 -\frac{\a}{\a-1}
 \frac{\left(\Tr \!\left\{A^{\a-1} B_Z\right\}\right) \left(\Tr \!\left\{A^{\a-1} B_Y\right\}\right)}{(\Tr A^\a)^2}.
\end{eqnarray*} 
Therefore
\begin{eqnarray*}
&&{\!\!\!\!\!\!\!\! D_\a((XYZ)_\r\|\s)}\\
&&{\!\!\!\!= 
\frac{1}{\a -1}
 \frac{\Tr \!\left\{(XYA^{\a-1}) B_Z\right\}+\Tr \!\left\{(YA^{\a-1})
(XB_Z)\right\}}{\Tr A^\a}}\\
&&
+\frac{1}{\a -1}
 \frac{\Tr\!\left\{(XA^{\a-1}) C_{YZ}\right\}+\Tr\!\left\{A^{\a-1} (XC_{YZ})\right\} }{\Tr A^\a}\\
&&
-\frac{1}{\a -1}
 \frac{\left(\Tr \!\left\{(YA^{\a-1}) B_Z\right\}+\Tr\!\left\{A^{\a-1} C_{YZ}\right\} \right)\Tr\{\a A^{\a-1}B_X\}}{(\Tr A^\a)^2}\\
 &&
-\frac{\a}{\a-1}
 \frac{\left(\Tr X\!\left\{A^{\a-1} B_Z\right\}\right) \left(\Tr \!\left\{A^{\a-1} B_Y\right\}\right)+\left(\Tr \!\left\{A^{\a-1} B_Z\right\}\right) \left(\Tr X\!\left\{A^{\a-1} B_Y\right\}\right)}{(\Tr A^\a)^2}\\
&&
+\frac{\a}{\a-1}
 \frac{\left(\Tr \!\left\{A^{\a-1} B_Z\right\}\right) \left(\Tr \!\left\{A^{\a-1} B_Y\right\}\right)
 \cdot 2\, \Tr\left\{\a A^{\a-1}B_X\right\}}{(\Tr A^\a)^3}.
\end{eqnarray*}
Since
\begin{eqnarray*}
 \left. \phantom{\prod\!\!\!\!\!\!\!\!\!} \Tr A^\a\right|_{\s=\r}=1,\qquad
 \left. \phantom{\prod\!\!\!\!\!\!\!\!\!} \Tr \left\{A^{\a-1}B_X\right\}\right|_{\s=\r}=0,\qquad
 \left. \phantom{\prod\!\!\!\!\!\!\!\!\!}\Tr \left\{A^{\a-1}C_{YZ}\right\}\right|_{\s=\r}=0,
\end{eqnarray*}
as in the proof of (\ref{eqn:fAlpha1}), and
\begin{eqnarray*}
\left. \phantom{\prod\!\!\!\!\!\!\!\!\!} \Tr\!\{A^{\a-1} (XC_{YZ})\}\right|_{\s=\r} = \Tr (XYZ\r)=XYZ \Tr\r=0,
\end{eqnarray*}
we have
\begin{eqnarray}
&&\left. \phantom{\prod\!\!\!\!\!\!\!\!\!}D_\alpha((XYZ)_\r \|\s) \right|_{\s=\r} \nonumber \\
&&\quad=
 \frac{1}{\a -1}\left.\phantom{\prod}\!\!\!\!\!\!\!
\Tr \!\left \{(XYA^{\a-1}) B_Z+(YA^{\a-1}) (XB_Z)+(XA^{\a-1}) C_{YZ} \right\}\right|_{\s=\r}.  \nonumber \\
\label{eqn:XYZD}
\end{eqnarray}

We next compute $X g^{(D_\a)}(Y,Z)$. We see from (\ref{eqn:fAlpha1}) that
\begin{eqnarray*}
g_\r^{(D_\a)}(Y,Z)
&\!\!=\!\!&
\left. \frac{1}{\a -1} \Tr \!\left\{(YA^{\a-1}) B_Z\right\} \right|_{\s=\r}\\
&\!\!=\!\!&
 \frac{1}{\a -1}
\Tr\left \{ (\r^{\frac{1-\a}{2\a}}(Z\r)\r ^ {\frac{1-\a}{2\a}})  (YA^{\a-1})\left.\phantom{\prod}\!\!\!\!\!\!\!\!\right|_{\s=\r}\right\}\\
&\!\!=\!\!&
 \frac{1}{\a -1}
\Tr\left\{ (Z\r) \left[\r^{\frac{1-\a}{2\a}} \left.\phantom{\prod}\!\!\!\!\!\!\!\! (YA^{\a-1}) \right|_{\s=\r} \r^{\frac{1-\a}{2\a}} \right]
\right\}.
\end{eqnarray*}
Using this identity, we have
\begin{eqnarray}
X g^{(D_\a)}(Y,Z)
&\!\!=\!\!&
\frac{1}{\a-1}
\Tr\left\{(XZ\r) \left[\r^{\frac{1-\a}{2\a}} \left.\phantom{\prod}\!\!\!\!\!\!\!\!(YA^{\a-1})\right|_{\s=\r} \r^{\frac{1-\a}{2\a}} \right]
 \right\} \nonumber \\
&&
+
\frac{1}{\a-1}
\Tr \left\{(Z\r)\, X\left[\r^{\frac{1-\a}{2\a}} \left.\phantom{\prod}\!\!\!\!\!\!\!\!(YA^{\a-1})\right|_{\s=\r}\r^{\frac{1-\a}{2\a}}\right]
 \right\} \nonumber \\
&\!\!=\!\!&
\frac{1}{\a-1}
\Tr\left\{\r^{\frac{1-\a}{2\a}} (XZ\r) \r^{\frac{1-\a}{2\a}} \left.\phantom{\prod}\!\!\!\!\!\!\!\!(YA^{\a-1})\right|_{\s=\r} 
 \right\} \nonumber \\
&&
+
\frac{1}{\a-1}
\Tr \left\{(Z\r)\, X\left[\r^{\frac{1-\a}{2\a}} \left.\phantom{\prod}\!\!\!\!\!\!\!\!(YA^{\a-1})\right|_{\s=\r}\r^{\frac{1-\a}{2\a}}\right]
 \right\} \nonumber \\
&\!\!=\!\!&
\frac{1}{\a-1}
 \left.\phantom{\prod}\!\!\!\!\!\!\!\! \Tr\left\{(X B_Z) (YA^{\a-1}) \right\} \right|_{\s=\r} \nonumber \\
&&
+
\frac{1}{\a-1}
\Tr \left\{(Z\r)\, X\left[\r^{\frac{1-\a}{2\a}} \left.\phantom{\prod}\!\!\!\!\!\!\!\!(YA^{\a-1})\right|_{\s=\r}\r^{\frac{1-\a}{2\a}}\right]
 \right\}. 
 \label{eqn:Xg}
\end{eqnarray}
Substituting (\ref{eqn:XYZD}) and (\ref{eqn:Xg}) into (\ref{eqn:nablaApp}) and (\ref{eqn:nablaStarApp}), we obtain the formulae (\ref{eqn:nabla}) and (\ref{eqn:nablaStar}).


\end{document}